\documentclass[conference]{IEEEtran}

\usepackage{bbm,dsfont,mathrsfs,fixmath}
\usepackage{amsmath,amssymb,eucal,graphicx}
\usepackage{stmaryrd}
\usepackage{amsmath,amstext,amsfonts,amssymb}
\usepackage{epsfig}
\usepackage{exscale}

\newtheorem{definition}{Definition}
\newtheorem{lemma}{Lemma}

\newtheorem{remark}{Remark}[section]

\newtheorem{theorem}{Theorem}

\newcommand{\cond}{\,\vert\,}

\newcommand{\defeq}{\triangleq}

\newcommand{\T}{{\scriptscriptstyle\mathsf{T}}}

\newfont{\bbb}{msbm10 scaled 500}

\newfont{\bb}{msbm10 scaled 1100}
\newcommand{\CC}{\mbox{\bb C}}

\newcommand{\gv}{{\bf g}}
\newcommand{\hv}{{\bf h}}

\newcommand{\rv}{{\bf r}}

\newcommand{\uv}{{\bf u}}

\newcommand{\vv}{{\bf v}}
\newcommand{\xv}{{\bf x}}
\newcommand{\yv}{{\bf y}}

\newcommand{\zerov}{{\bf 0}}

\newcommand{\Am}{{\bf A}}

\newcommand{\Cc}{{\cal C}}

\newcommand{\Nc}{{\cal N}}

\newcommand{\trace}{{\hbox{tr}}}

\DeclareFontFamily{U}{cmfi}{}
\DeclareFontShape{U}{cmfi}{m}{n}{ <-> cmfi10 }{}
\DeclareSymbolFont{CMFI}{U}{cmfi}{m}{n}

\newcommand{\vh}{\pmb{h}}
\newcommand{\vg}{\pmb{g}}
\newcommand{\vx}{\pmb{x}}

\newcommand{\vu}{\pmb{u}}
\renewcommand{\vv}{\pmb{v}}

\newcommand{\rvh}{H}
\newcommand{\rvg}{G}

\newcommand{\rvv}{V}

\newcommand{\rz}{Z}
\newcommand{\ry}{Y}

\renewcommand{\rv}{V}

\newcommand{\taulog}{\mathcal{O}_P}

\begin{document}

\title{On the Secrecy Degrees of Freedom of Multi-\\Antenna Wiretap Channels with Delayed CSIT}
\thanks{This work was partially supported by the framework of the FP7 Network of Excellence in Wireless Communications NEWCOM++.}
\author{\authorblockN{Sheng Yang, Pablo Piantanida, Mari Kobayashi}
\authorblockA{
SUPELEC \\
Gif-sur-Yvette, France\\
 {\tt \{sheng.yang,pablo.piantanida,mari.kobayashi\}@supelec.fr}
}
\and
\authorblockN{Shlomo Shamai (Shitz)}
\authorblockA{ Technion-Israel Institute of Technology\\
Haifa, Israel\\  
{\tt sshlomo@ee.technion.ac.il}
}
}
\maketitle
\vspace{-25mm}

\begin{abstract}
The secrecy degrees of freedom (SDoF) of the Gaussian multiple-input and single-output (MISO)
wiretap channel is studied under the assumption that
delayed channel state information (CSI) is available at the
transmitter and each receiver knows its own instantaneous channel.
We first show that
a strictly positive SDoF can be guaranteed whenever the transmitter has delayed CSI (either on the legitimate channel or/and the eavesdropper channel). In particular, in the case with delayed CSI on
both channels, it is shown that the optimal SDoF is
$2/3$. We then generalize the result to the two-user Gaussian MISO broadcast channel with confidential messages
and characterize the SDoF region when the transmitter has delayed CSI of both receivers. Interestingly, the artificial noise schemes exploiting several time instances are shown to provide the optimal SDoF region by masking the confidential message to the unintended receiver while aligning the interference at each receiver. 
\end{abstract}

\section{Introduction}
Although perfect channel state information at transmitter (CSIT) may not be available in most practical scenarios due to time-varying nature
of wireless channels, many wireless applications must still guarantee
secure and reliable communication. In fast fading scenarios, the channel
estimation/feedback process is often slower than the coherence time and
CSIT may be further outdated. In \cite{maddah2010degrees}, the authors
considered such a scenario in the context of multi-input single-output
(MISO) broadcast channels~(BCs). By assuming delayed CSIT from each
receiver and perfect CSI at the receivers, they established the optimal
sum-degrees of freedom~(DoF). These results show that, by a careful
design of linear precoding schemes, completely outdated CSIT, i.e.
independent of the current channel state, can still significantly
increase the DoF. Recently, \cite{vaze2010degrees} extended the work for
two-user multi-input multi-output (MIMO) BCs and characterized the DoF
region.  The same feedback model has also been studied in
\cite{maleki2010retrospective} where the so-called retrospective
interference alignment has been proposed for networks with distributed
encoders~(e.g. interference channels and X-channels). Finally,
\cite{vaze2011degrees} established the DoF region of the two-user MIMO
interference channel.

The secrecy capacity of MISO Gaussian wiretap channel is not  fully
understood yet for the cases of partial (or imperfect) CSIT.  Due to the
difficulty of its complete
characterization, a number of contributions have focused on secrecy degrees of freedom (SDoF) capturing the behavior in high signal-to-noise (SNR) regime~(see e.g.
\cite{yingbin2009compound,khisti2010compound,kobayashi2009compound,kobayashi2010secrecy}). References \cite{yingbin2009compound,khisti2010compound,
kobayashi2009compound} investigated compound models where the channel uncertainty is modeled as a set of finite states, while \cite{kobayashi2010secrecy} considered the case when the transmitter knows some special structure of the block-fading channels of receivers. 

In this paper, inspired by recent exciting results, we
study the impact of delayed CSIT on the MISO Gaussian wiretap and the MISO Gaussian BC
with confidential messages (BCC). We assume that delayed CSI is available both at the transmitter and at the receivers (or
eavesdroppers), where each receiver knows its own instantaneous channel.  
We consider two different cases for the wiretap channel: (i) the ``asymmetric scenario'' where the transmitter has delayed CSI of either the
legitimate channel or the eavesdropper channel, and (ii) the ``symmetric scenario'' where the transmitter has delayed CSI of both channels. 
It is shown that, similarly to the conclusion drawn in
\cite{maddah2010degrees,maleki2010retrospective}, delayed CSIT can
increase the SDoF. More precisely, by means of simple artificial noise
schemes, a SDoF of $1/2$ can be guaranteed in the asymmetric scenario
while a SDoF of $2/3$ is ensured in the symmetric case. It turns out
that $2/3$ is the fundamental SDoF for symmetric scenario. Then, we consider the MISO BCC where the transmitter wishes to send two messages respectively to two receivers while keeping each of them secret to the unintended receiver. We characterize the SDoF region and show that the artificial noise to convey two messages enables to achieve the sum rate SDoF point $(\frac{1}{2},\frac{1}{2})$. 

The rest of the paper is organized as follows. 
Section II presents the system model while Section III provides an upper bound and artificial
noise schemes for the wiretap channel. The SDoF
region of MISO-BCC are derived in Section IV. Finally Section V concludes the paper. 
We should emphasize that all the results of this work apply for 
$M\geq 2$, although the achievability results are provided for $M=2$ for the
sake of simplicity. 

{\bf Notations:} Upper case letters, lower case bold letters are used to
denote random variables, vectors, respectively. $X^n$ denotes the
sequence $(X_1,\ldots,X_n)$.  $\Am^\T$ and $\trace ( \Am)$
denote the transpose and the trace of matrix $\Am$, respectively.
$h(X)$ denotes the differential entropy of random variable $X$.
$\taulog$ denotes any real-valued function $f(P)$ such that
${\displaystyle \lim_{P\to \infty}} \frac{f(P)}{\log_2 P}=0$. 


\section{System Model}
Consider the fading Gaussian MISO wiretap channel, where the transmitter with $M$ antennas sends a confidential message to the legitimate receiver in the presence of an eavesdropper. The corresponding channel models are given by
\begin{align*}
    y_{t} = \vh_{t}^{\T} \vx_{t} + e_{t}, \\
    z_{t} = \vg_{t}^{\T}  \vx_{t} + b_{t},
\end{align*}
for $t=1,\ldots,n$, where $(y_{t}, z_{t})$ denotes the observations at the legitimate receiver and the eavesdropper at channel use $t$, associated to $M$-input single-output channel vector $\vh_{t},\vg_{t}\in\CC^{M\times 1}$, respectively, and $(e_{t},b_{t})$ are assumed to be independent and identically distributed (i.i.d.) additive white Gaussian noises $\sim\Nc_{\Cc}(0,1)$, the input vector $\vx_{t}\in\CC^{M\times 1}$ is subject to the power constraint 
$
\frac{1}{n}\sum\limits_{t=1}^n \trace ( \vx_{t}  \vx_{t}^H  ) \leq P.
$  
We assume any arbitrary stationary fading process where $\{\vh_{t}, \vg_{t}\}_{t=1}^\infty$ are mutually independent and change from a letter to another one in an independent manner.\vspace{1mm}

\begin{definition} A code for the Gaussian MISO wiretap channel with delayed CSI consists of:
\begin{itemize}
\item A sequence of stochastic encoders\footnote{If delayed CSI from  one terminal is available, the encoder depends only on $\{\vg_{1},\dots,\vg_{t-1}\}$ or $\{\vh_{1},\dots,\vh_{t-1}\}$.} $F_t:\{1,\dots,M_n\}\times \{\vh_{1},\dots,\vh_{t-1}\} \times \{\vg_{1},\dots,\vg_{t-1}\}  \longmapsto \CC^{M}$ where the message $W$ is uniformly distributed over $\{1,\dots,M_n\}$, 
\item A legitimate decoder  given by the mapping $\hat{W}:\{y_{1},\dots,y_{n}\}\times \{\vh_{1},\dots,\vh_{n}\} \longmapsto \{1,\dots,M_n\}$,
\item The error probability is then defined by 
$$
P_e^{(n)}=\Pr\left\{ \mathbf{W}\neq \mathbf{\hat{W}}\right\}. 
$$
\end{itemize}\vspace{1mm}

An SDoF $d\geq 0$ is said {\em achievable} if there exists a code that simultaneously satisfies
\begin{eqnarray*}
\lim_{P\rightarrow \infty} \liminf_{n\rightarrow \infty} \frac{n^{-1} \log_2 M_n(P)}{\log_2 P}\geq d,
\end{eqnarray*}
with 
\begin{eqnarray*}
\lim\limits_{n\rightarrow \infty} P_{e}^{(n)} = 0,
\end{eqnarray*}
and the equivocation 
\begin{eqnarray*}\label{eq:SecurityConstraints}
\lim_{P\rightarrow \infty} \limsup_{n\rightarrow \infty} \frac{n^{-1} I(W;Z^n,H^n,G^n)}{\log_2 P}=0. 
 \end{eqnarray*}
The supremum of all achievable SDoF is then called the fundamental SDoF of the wiretap channel. 
\end{definition}

\section{MISO Wiretap Channel with Delayed CSIT}

The SDoF of the Gaussian MISO wiretap channel is upper-bounded by $1$, which is the DoF of a MISO channel. It is achievable when instantaneous CSI on either the legitimate or the eavesdropper channel is available at the transmitter. In this section, we first provide a new upper bound on the SDoF when no instantaneous CSI is available at the transmitter. It will then be shown that this upper-bound is achievable for the symmetric scenario where delayed CSIT from both the legitimate and eavesdropper channel is available. 

\subsection{Upper Bound on the SDoF}

\begin{theorem}[upper bound] \label{thm:ub}
 Without instantaneous CSIT, the SDoF of the Gaussian MISO wiretap channel is upper-bounded by $d=\frac{2}{3}$. 
\end{theorem}

For sake of clarity, we remove the channel state from the expressions since these can be considered as additional channel outputs. Before proving the Theorem, let us start by setting the following constraints: 
\begin{align}
  h(\ry_t \cond \ry^{t-1}, \rz^{t-1}) &= h(\rz_t \cond \ry^{t-1}, \rz^{t-1}),
  \label{eq:tmp01}\\
  h(\ry_t \cond \ry^{t-1}, \rz^{t-1}, W) &= h(\rz_t \cond \ry^{t-1}, \rz^{t-1}, W),
  \label{eq:tmp02}
\end{align}%
for $t=1,\dots,n$.  Note that these are direct consequences of our assumptions: (i) the legitimate and the eavesdropper channel have the same statistics, (ii) the channel input cannot depend on either of the instantaneous channels, and (iii) the marginal distributions of both outputs are equal given the same previous observations and/or the source message. \vspace{1mm}
\begin{lemma}
The following inequalities hold true under the constraints
\eqref{eq:tmp01} and \eqref{eq:tmp02}:  
  \begin{align}
    h(\rz^n) &\ge h(\ry^n \cond \rz^n), \label{eq:tmp1}\\
    h(\ry^n) &\ge h(\rz^n \cond \ry^n),\\
    h(\rz^n\cond W) &\ge h(\ry^n \cond \rz^n, W), \\
    h(\ry^n\cond W) &\ge h(\rz^n \cond \ry^n, W). 
  \end{align}%
\end{lemma}\vspace{1mm}
\begin{proof}
  By symmetry of the problem, it is enough to prove the first inequality \eqref{eq:tmp1} as follows 
  \begin{align}
   2h(\rz^n)&= 2\sum_{t=1}^n h(\rz_t \cond
    \rz^{t-1}) \label{eq:cr1}\\
    &\ge 2\sum_{t=1}^n h(\rz_t \cond \ry^{t-1},
    \rz^{t-1})\label{eq:cond} \\
    &= \sum_{t=1}^n  h(\ry_t \cond \ry^{t-1}, \rz^{t-1}) +
    h(\rz_t \cond \ry^{t-1}, \rz^{t-1}) \label{eq:tmp343}\\
    &\ge \sum_{t=1}^n h(\ry_t, \rz_t \cond \ry^{t-1},
    \rz^{t-1}) \label{eq:tmp874} \\
    &= h(\ry^{n}, \rz^{n}) \label{eq:tmp291}\\
    &= h(\rz^{n}) + h(\ry^{n} \cond \rz^{n})
    \label{eq:tmp754}
  \end{align}%
  where \eqref{eq:cr1} and \eqref{eq:tmp291} are from the chain rule;
  \eqref{eq:cond} holds because conditioning reduces entropy;
  \eqref{eq:tmp343} is from \eqref{eq:tmp01}. From \eqref{eq:tmp754}, \eqref{eq:tmp1} is immediate. 
\end{proof}

We are now ready to provide the following lemma that is essential to our
main results. 
\begin{lemma}
  Under constraints \eqref{eq:tmp01} and \eqref{eq:tmp02}, we have:
  \begin{align}
    h(\ry^n) &\le 2 h(\rz^n), \label{eq:tmp822}\\ 
    h(\rz^n) &\le 2 h(\ry^n), \label{eq:tmp823}\\ 
    h(\ry^n\cond W) &\le 2 h(\rz^n \cond W), \label{eq:tmp824}\\ 
    h(\rz^n\cond W) &\le 2 h(\ry^n \cond W),\label{eq:tmp876} \\
    I(W;\ry^n) - I(W;\rz^n) &\le h(\rz^n).\label{eq:tmp825} 
  \end{align}%
\end{lemma}
\begin{proof}
  To prove \eqref{eq:tmp822}, from \eqref{eq:tmp291}, we have
  \begin{align*}
    2 h(\rz^n) &\ge h(\ry^n, \rz^n) \\ 
    &= h(\ry^n) + h(\rz^n \cond \ry^n) \\
    &\ge h(\ry^n)  
  \end{align*}%
  where the last inequality\footnote{It is true since $\rz^n$ contains
  AWGN that is independent from $\ry^n$.} comes from the fact that $h(\rz^n \cond
  \ry^n)\geq \taulog$. Same steps can be applied to obtain
  \eqref{eq:tmp823}-\eqref{eq:tmp824}.  To show \eqref{eq:tmp825},
  we start from \eqref{eq:tmp1}
\begin{align*}
  h(\rz^n) &\ge h(\ry^n \cond \rz^n) \\
  &\ge  I(W;\ry^n \cond \rz^n) \\
  &\ge  I(W;\ry^n \cond \rz^n) -  I(W;\rz^n \cond \ry^n) \\
  &= I(W;\ry^n) - I(W;\rz^n). 
\end{align*}%
\end{proof}
The inequality (\ref{eq:tmp876}) implies that 
\begin{align}
  \lefteqn{I(W;\ry^n) - I(W;\rz^n)} \nonumber\qquad \\ &= h(\ry^n) -
  h(\ry^n\cond W) - h(\rz^n) + h(\rz^n\cond W) \nonumber \\
  &\le  h(\ry^n) + \frac{1}{2} h(\rz^n\cond W) - h(\rz^n) \nonumber \\
  &\le h(\ry^n)  - \frac{1}{2} h(\rz^n)\label{eq:tmplast}
\end{align}
By combining two bounds (\ref{eq:tmp876}) and (\ref{eq:tmplast}), we have
\begin{align}
  \lefteqn{I(W;\ry^n) - I(W;\rz^n)} \nonumber\qquad \\ &\le \min\left\{ h(\rz^n),\ h(\ry^n) - \frac{1}{2}
  h(\rz^n) \right\} \nonumber\\
  &\le\max_{h(\ry^n)} \max_{h(Z^n)} \min\left\{  h(\rz^n),\ h(\ry^n) - \frac{1}{2}
  h(\rz^n) 
  \right\} \\
  &\le \frac{2}{3} n \log_2(P) + \taulog. 
\end{align}%
We now verify that \eqref{eq:tmp01} and \eqref{eq:tmp02} still hold 
given $\rvh^n$ and $\rvg^n$
\begin{align*}
  \lefteqn{h(\ry_t \cond \ry^{t-1}, \rz^{t-1}, \rvh^n, \rvg^n)}\qquad\\ &= h(\ry_t \cond \ry^{t-1},
  \rz^{t-1}, \rvh^{t-1}, \rvg^{t-1}, \rvh_t) \\
  &= h(\rz_t \cond \ry^{t-1}, \rz^{t-1}, \rvh^{t-1}, \rvg^{t-1}, \rvg_t) \\
  &= h(\rz_t \cond \ry^{t-1}, \rz^{t-1}, \rvh^n, \rvg^n)  
\end{align*}%
from the fact that current channel outputs do not depend on the future channel
realizations. Similarly, 
\begin{multline*}
  h(\ry_t \cond \ry^{t-1}, \rz^{t-1}, \rvh^n, \rvg^n, W) \\= h(\rz_t \cond \ry^{t-1}, \rz^{t-1}, \rvh^n, \rvg^n, W).  
\end{multline*}

We are ready to prove Theorem \ref{thm:ub} as follows. From Fano's inequality and the secrecy constraint we have that
\begin{align*}
  n(R&-\taulog ) \\
&\le I(W;\ry^n \cond \rvh^n, \rvg^n) - I(W;\rz^n \cond \rvh^n, \rvg^n) \\
    &\le \min\left\{ h(\rz^n \cond \rvg^n ),\ h(\ry^n \cond \rvh^n ) - \frac{1}{2}
    h(\rz^n \cond \rvg^n ) \right\}  \\
  &\le \frac{2}{3} n \log_2 P,
\end{align*}%
which concludes the proof of the theorem. 

\subsection{Achievability: Symmetric Case}
With delayed CSIT on both the legitimate and eavesdropper channels, the
upper bound is indeed achievable. 
\begin{theorem}[symmetric case]
  The fundamental SDoF of a two-user MISO wiretap channel with delayed CSIT from both the legitimate
  and the eavesdropper channel is $d=\frac{2}{3}$. \vspace{1mm}
\end{theorem}
\begin{proof}
The converse follows from Theorem \ref{thm:ub}. Inspired by the artificial noise (AN) scheme~\cite{goel2008guaranteeing}, we
propose a three-slot scheme sending four independent
Gaussian-distributed symbols $\vu \defeq [u_1\ u_2]^\T$, $\vv\defeq
[v_1\ v_2]^\T$,  whose powers scale equally with $P$.  In the first slot, the AN $\uv$ is sent. In the second slot, the transmitter sends the useful symbols $\vv$ together with the AN seen by the legitimate receiver in the first slot. Finally, we repeat the observation of the eavesdropper in the second slot (without thermal noise). By ignoring scaling terms in the transmit vectors, the channel inputs/outputs are given in \eqref{eq:tmp201}-\eqref{eq:tmp203}.
\begin{figure*}[!th]
\begin{flalign}
   \xv_1 &= \vu& 
   \xv_2 &= \vv+ [ \hv_1^{\T} \uv \quad 0]^\T
          &  
   \xv_3 &=[\gv_2 ^\T \vv + g_{21}\hv_1^{\T} \uv\quad 0]^\T
   \label{eq:tmp201}\\
   y_1 &= \hv_1^\T \uv + e_1 & 
   y_2 &= \hv_2^{\T}\vv + h_{21}\hv_1^{\T} \uv + e_2 & 
   y_3 &=  h_{31} \gv_2 ^\T \vv +  h_{31}g_{21}\hv_1^{\T} \uv + e_3 \\
   z_1 &= \gv_1^{\T} \uv + b_1 & 
   z_2 &= \gv_2 ^\T \vv + g_{21}\hv_1^{\T} \uv + b_2 &
   z_3 &= g_{31} \gv_2 ^\T \vv +g_{31} g_{21}\hv_1^{\T} \uv+ b_3
   \label{eq:tmp203} 
\end{flalign}
\end{figure*}
The received signals can be rewritten as 
\begin{align}
  \begin{bmatrix} y_1 \\ y_2 \\ y_3\end{bmatrix} &= 
    \begin{bmatrix} 
      1  & \zerov \\
       h_{21} & \hv_2^\T\\ 
        h_{31}g_{21} & h_{31}\gv_2^{\T}
  \end{bmatrix}_{3\times3} \begin{bmatrix}  \hv_1^{\T}\uv \\
          \vv\end{bmatrix}_{3\times1} +
    \begin{bmatrix} e_1 \\ e_2\\ e_3\end{bmatrix}, \\
  \begin{bmatrix} z_1 \\ z_2\\ z_3 \end{bmatrix} &= \begin{bmatrix}
   \gv_1^{\T}  & 0 \\
          g_{21}\hv_1^{\T} & 1 \\
          g_{31}g_{21}\hv_1^{\T} & g_{31}
   \end{bmatrix}_{3\times3} \begin{bmatrix} \uv \\
          \gv_2^{\T}\vv\end{bmatrix}_{3\times1} + \begin{bmatrix} b_1 \\ b_2\\ b_3\end{bmatrix}.
\end{align}%
The following remarks are in order; (i) since the equivalent channel matrix is full-rank, the useful signal $\vv$ can be recovered from $\yv$,  (ii) the eavesdropper's observation is completely drowned in the artificial noise $\vu$. More precisely, we have
\begin{align}
  I(\rvv;\ry^3) &= 2 \log_2(P) + \taulog, \\
  I(\rvv;\rz^3) &= \taulog,
\end{align}%
which implies a SDoF $d=2/3$. 
\end{proof}
\subsection{Achievability: Asymmetric Case}
\begin{theorem}(asymmetric case)
  With delayed CSIT only on the legitimate channel, an SDoF $d=1/2$ is
  achievable. 
\end{theorem}
\begin{proof}
The achievability is based on the following two-slot scheme sending three independent
Gaussian-distributed symbols $\vu \defeq [u_1\ u_2]^\T$ and $v$: 
\begin{flalign}
   \vx_1 &= \vu &
   \vx_2 &= [\vh_1^\T \vu\quad  v]^\T  \\
   y_1 &= \vh_1^\T \vu + e_1 & 
   y_2 &= h_{21} (\vh_1^\T \vu) + h_{22}  v + e_2  \\
   z_1 &= \vg_1^\T \vu + b_1 & 
   z_2 &= g_{21} (\vh_1^\T \vu) + g_{22}  v + b_2 
\end{flalign}
The received signal can be rewritten as 
\begin{align}
  \begin{bmatrix} y_1 \\ y_2 \end{bmatrix} &= 
    \begin{bmatrix} 
      1 & 0 \\ h_{21} & h_{22} \end{bmatrix}_{2\times2} \begin{bmatrix}
        \vh_1^\T \vu \\ v \end{bmatrix}_{2\times1} +
    \begin{bmatrix} e_1 \\ e_2\end{bmatrix} \\
  \begin{bmatrix} z_1 \\ z_2 \end{bmatrix} &= \begin{bmatrix}
    \vg_1^\T & 0 \\ g_{21} \vh_1^\T & g_{22}
  \end{bmatrix}_{2\times3}\begin{bmatrix} \vu \\ v
  \end{bmatrix}_{3\times1}  + \begin{bmatrix} b_1 \\ b_2\end{bmatrix} 
\end{align}%
from which we remark that: (i) the useful signal $v$ can be recovered from $\yv$, and (ii) it is completely drowned in the artificial
noise $\vu$ at the eavesdropper side, i.e., 
\begin{align}
  I(\rv;\ry^2) &= \log_2(P) + \taulog, \\
  I(\rv;\rz^2) &= \taulog,
\end{align}%
which implies an SDoF $d=1/2$. 
\end{proof}\vspace{1mm}

It is still unknown if $1/2$ is the best possible SDoF with only delayed
CSIT on the legitimate channel. Nevertheless, it can be shown that it is
indeed optimal within the class of Gaussian inputs. As a matter of fact,
we can show that  
\begin{align}
  h(\rz^n \cond \rvh^n, \rvg^n) \ge h(\ry^n \cond \rvh^n, \rvg^n) + \taulog
\end{align}%
the proof of which is omit due to page limit. Therefore, it is straightforward to get
\begin{align*}
  n(R-\taulog) 
  &\le I(W;\ry^n \cond \rvh^n, \rvg^n) - I(W;\rz^n \cond \rvh^n, \rvg^n) \\
  &\le h(\rz^n \cond W, \rvh^n, \rvg^n) - h(\ry^n \cond W, \rvh^n, \rvg^n)  \\
  &\le  \frac{1}{2} h(\rz^n \cond W, \rvh^n, \rvg^n) \\
  &\le \frac{1}{2} n \log_2 P.   
\end{align*}

\section{Broadcast Channel with Confidential Messages (BCC)}

We now characterize the fundamental SDoF region of the two-user MISO-BCC with delayed CSIT on both channels.  In this setting, the transmitter wishes to send two messages $(W_1,W_2)$ to receivers 1 and 2, respectively, while keeping each of them secret to the unintended receiver, i.e.
\begin{align}  \label{eq:Constraint1-bcc}
& \lim_{P\rightarrow \infty} \limsup_{n\rightarrow \infty} \frac{n^{-1} I(W_1;Z^n,H^n,G^n)}{\log_2 P}=0,\\  \label{eq:Constraint2-bcc}
& \lim_{P\rightarrow \infty} \limsup_{n\rightarrow \infty} \frac{n^{-1} I(W_2;Y^n,H^n,G^n)}{\log_2 P}=0. 
\end{align}
The channel models, the definition of a code and achievability remain similar to those of Section II. 
Let us begin with the proof of the outer bound on the SDoF region. Then, we show the achievability of the corner (sum SDoF) point involved in the region and by a simple time-sharing argument we will prove that our outer bound is tight. 

\subsection{Outer Bound on the SDoF Region of BCC}

\begin{theorem}[outer bound] \label{thm:ubBCC}
The SDoF region of the two-user MISO-BCC with delayed CSIT from both receivers is outer-bounded by
\begin{align*}
\mathcal{R}_{\textrm{BCC}} = \Big\{(d_1,d_2)\in\mathbb{R}_2^+:\quad
3d_1+d_2 \leq 2,\ d_1+3d_2 \leq 2 \Big\}.
\end{align*}
The region is illustrated in Fig. 1.
\end{theorem}

\begin{figure}[ht]\label{fig:Region}
\begin{center}
\includegraphics[width=5cm,clip=]{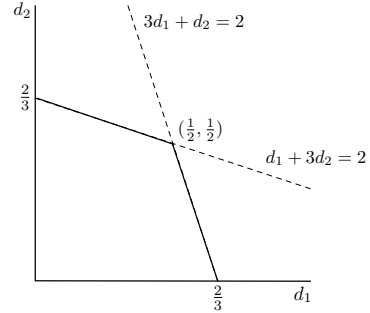}
\caption{The SDoF region of the two-user MISO-BCC.}
\end{center}
\end{figure}

\begin{remark}
Obviously, the above region is included in that of the MISO-BC with delayed CSIT \cite{maddah2010degrees} as well as the rectangle region of the MISO-BCC with perfect CSIT. In Table \ref{table:comparison}, we summarize the achievable SDoF with no, delayed, and perfect CSIT and compare them the achievable DoF of the two-user MISO-BC. We remark that the degradation due to the imperfect CSIT is more significant in the secrecy communications.   
 \end{remark}
\begin{table}[!bhtp]
\caption{Sum SDoF of the MISO-BCC. }
\begin{center}
\begin{tabular}{|c|c|c|c|}
\hline &  no CSIT & delayed CSIT & perfect CSIT\\ \hline\hline
MISO-BCC &  0  & 1 & 2 \\\hline
MISO-BC & 1 & $\frac{4}{3}$ &  2 \\\hline
\end{tabular}
\end{center}
\label{table:comparison}
\vspace{-0.2em}
\end{table}

\begin{figure*}[!ht]
\begin{flalign}
   \xv_1 &= \vu& 
   \xv_2 &=  \vv_1+ [  \hv_1^{\T} \uv\quad 0]^\T &  
   \xv_3 &= \vv_2+ [\gv_1^{\T} \uv\quad 0]^\T&
      \xv_4 &= [  \gv_2 ^\T \vv_1 + \hv_3 \vv_2 + (h_{31} \gv_1^{\T} + g_{21}\hv_1^{\T} )\uv \quad 0]^\T     \label{eq:tmp401}\\  
   \tilde{y}_1 &= \hv_1^\T \uv & 
   \tilde{y}_2 &= \hv_2^{\T}\vv_1 + h_{21}\hv_1^{\T} \uv & 
   \tilde{y}_3 &= \hv_3^\T \vv_2 + h_{31} \gv_1^{\T} \uv &
   \tilde{y}_4 &= h_{41}\left(\gv_2 ^\T \vv_1 + \hv_3 ^\T\vv_2 + (h_{31} \gv_1 + g_{21}\hv_1 )^\T \uv\right) \label{eq:tmp402} \\
   \tilde{z}_1 &= \gv_1^{\T} \uv  & 
   \tilde{z}_2 &= \gv_2 ^\T \vv_1 + g_{21}\hv_1^{\T} \uv  &
   \tilde{z}_3 &= \gv_3 ^\T\vv_2 + g_{31}\gv_1^{\T} \uv &
   \tilde{z}_4 &=g_{41}\left(\gv_2 ^\T \vv_1 + \hv_3^\T \vv_2 + (h_{31} \gv_1+ g_{21}\hv_1)^\T\uv\right) \label{eq:tmp403} 
\end{flalign}
\end{figure*}
\vspace{-0.2em}
\begin{proof}
First, the secrecy constraint (\ref{eq:Constraint1-bcc}) and the Fano inequality for $W_2$ yield 
\begin{equation} 
I(W_1;\rz^n|W_2)\leq n \taulog \label{eq:ubBCC6}.
\end{equation}
Combining \eqref{eq:ubBCC6} with the Fano inequality on $W_1$, we have
\begin{align}
n(R_1- \taulog) &\leq  I(W_1;Y^n|W_2) - I(W_1;Z^n|W_2)\nonumber\\
&\leq  I(W_1;Y^n|Z^n,W_2)\nonumber \\
 &\leq h(Y^n|Z^n,W_2) \label{apply-lemma}\\
&\leq h(Z^n|W_2)
\end{align}
where (\ref{apply-lemma}) follows from inequality \eqref{eq:tmp1} in Lemma 1 and the last inequality follows since removing the conditioning increases the entropy. We notice that the upper bound (\ref{eq:tmplast}) holds true by replacing $W$ with $W_1$. Finally, we obtain the following
upper bound for $R_1$: 
\begin{multline}\label{eq:ubBCC7}
 n(R_1- \taulog) \leq  \min\{ h(Z^n|W_2), h(Y^n)-\frac{1}{2}h(Z^n)\}. 
\end{multline} 

On the other hand, the Fano inequality for $W_2$ leads to
\begin{equation} 
n(R_2- \taulog) \leq  h(\rz^n) - h(\rz^n|W_2), \label{eq:ubBCC8}
\end{equation}
By weight-summing the two inequalities (\ref{eq:ubBCC7}) and (\ref{eq:ubBCC8}), we obtain
\begin{align*}
n(3R_1+R_2- \taulog)
&\leq \max_{h(Y^n)}\max_{\alpha}\min\left\{ \alpha,\
3h(Y^n)-\frac{\alpha}{2} \right\}\\
&\leq \max_{h(Y^n)} 2h(Y^n)\\
& \leq 2n\log_2 P
\end{align*}
where we let $\alpha= h(Z^n)+2h(Z^n|W_2)$ in the first inequality and the last inequality follows because $h(Y^n)\leq n\log_2 P+\taulog$. 
By dividing both sides by $\log_2 P$ and letting $P$ grow, we obtain the first desired inequality. By swapping the roles of $R_1$ and $R_2$, we obtain the second inequality. This completes the proof.  
\end{proof}

It turns out that the outer bound given by Theorem \ref{thm:ubBCC} is the fundamental SDoF region of the two-user MISO-BCC with delayed CSIT on both channels. We next prove that the cross point between two half-spaces is indeed achievable and hence by the simple time-sharing argument all pairs $(d_1,d_2)\in \mathcal{R}_{\textrm{BCC}}$ are achievable.

\subsection{Achieving $(d_1, d_2) = (\frac{1}{2},\frac{1}{2})$}

As an extension of the three-slot scheme for the MISO wiretap channel, 
we propose a four-slot scheme sending six independent
Gaussian-distributed symbols $\vu \defeq [u_1\ u_2]^\T$, $\vv_1 \defeq
[v_{11} \ v_{12}]^\T$, $\vv_2 \defeq [v_{21} \ v_{22}]^\T$ whose powers scale equally with $P$. The channel inputs and outputs are given in \eqref{eq:tmp401}-\eqref{eq:tmp403}, where we let $\tilde{y}_t, \tilde{z}_t$ denote the received signal without thermal noise at receiver 1, 2, respectively. Compared to the three-slot scheme for the MISO wiretap channel,an additional time slot (third slot) is added to convey the message $W_2$ and the last slot is dedicated to send two signals overheard at the unintended receiver simultaneously.  The observations at two receivers can be rewritten as
\begin{align*}\label{obs-delayed2}
\left[\begin{array}{c}
            \tilde{y}_1\\
            \tilde{y}_2\\
            \tilde{y}_3\\
            \tilde{y}_4
          \end{array}\right] = &\left[\begin{array}{ccc}
           \zerov & 1 & 0\\
          \hv_2^{\T}  & h_{21} & 0 \\
          \zerov & 0 & 1 \\
          h_{41}\gv_2^{\T}&  h_{41}g_{21} &  h_{42} 
          \end{array}\right]           
          \left[\begin{array}{c}
          \vv_1\\
          \hv_1^{\T}\uv \\
          h_{31}\gv_1^{\T}\uv + \hv_3^{\T} \vv_2
           \end{array}\right],   \\
\left[\begin{array}{c}
            \tilde{z}_1\\
            \tilde{z}_2\\
            \tilde{z}_3\\
            \tilde{z}_4
          \end{array}\right] = &\left[\begin{array}{ccc}
           \zerov & 1 & 0\\
          \zerov  & 0 & 1 \\
          \gv_3^{\T} & g_{31} &0 \\
          g_{42}\hv_3^{\T}&  g_{42}h_{31} &  g_{41} 
          \end{array}\right]           
          \left[\begin{array}{c}
          \vv_2\\
          \gv_1^{\T}\uv \\
           \gv_2^{\T}\vv_1 + g_{21}\hv_1^{\T}\uv 
           \end{array}\right].              
\end{align*}
We remark that 1) $\vv_2$~(resp.
$\vv_1$) and the artificial noise $\vu$ are aligned in a
two-dimensional subspace at receiver $1$~(resp.~receiver~$2$), while the useful signal
$\vv_1$~(resp. $\vv_2$) also lies within a two-dimensional subspace, 2)
$\vv_2$~(resp. $\vv_1$) is completely drowned in the artificial noise
$\vu$ at receiver~$1$~(resp.~receiver~$1$). It is readily shown that 
\begin{align}
  I(\rvv_1;\ry^4) &= 2 \log_2(P) + \taulog \\
  I(\rvv_1;\rz^4\cond \rvv_2) &= \taulog
\end{align}%
which implies degrees of freedom $d_1 = 1/2$.  By symmetry, we have $d_2 =
1/2$. 


\section{Conclusions}
We studied the impact of delayed CSIT on secrecy degrees of freedom (SDoF) in the Gaussian MISO wiretap channel and the two-user Gaussian MISO broadcast channel with confidential messages (BCC). We fully characterized the corresponding SDoF region when the transmitter has delayed CSI on both channels and proved that simple artificial noise schemes are optimal. The comparison with the achievable DoF of the MISO-BC  demonstrated the sensitivity of the secrecy rate to the quality of CSIT. On one hand, delayed CSIT substantially increases the SDoF (w.r.t. the case of no CSIT where the SDoF is zero). On the other hand,  the lack of perfect CSIT yields a more severe loss in the secrecy communications. 




\section*{Acknowledgment}
This work was partially supported by the framework of the FP7 Network of Excellence in Wireless Communications NEWCOM++.

%

\end{document}